\documentclass[a4paper,twoside]{article}
\usepackage{amsmath,amssymb,amsthm}
\usepackage{amscd}
\usepackage[usenames,dvipsnames]{color}
\usepackage{pstricks}
\usepackage{graphicx}
\usepackage[all]{xy}
\usepackage{tabularx}

\newtheorem{prop}{Proposition}[section]

\newtheorem{defi}{Definition}[section]
\newtheorem{lemm}{Lemma}[section]
\newtheorem{thm}{Theorem}[section]
\newtheorem{coro}{Corollary}[section]
\newtheorem{ex}{Example}[section]

\begin{document}
\date{}
%
\title{\textbf{Extension of distributions, scalings and renormalization of QFT
on Riemannian manifolds.}}
\author{Dang Nguyen Viet\\
Laboratoire Paul Painlev\'e (U.M.R. CNRS 8524)\\
UFR de Math\'ematiques\\
Universit\'e de Lille 1\\
59 655 Villeneuve d'Ascq C\'edex France.}

\maketitle

\begin{abstract}
Let $M$ be a smooth manifold 
and $X\subset M$
a closed subset of $M$.
In this paper, we 
introduce a natural condition
of \emph{moderate growth}
along $X$ for a distribution $t$
in $\mathcal{D}^\prime(M\setminus X)$
and prove that this condition 
is equivalent to the existence of an extension
of $t$
in $\mathcal{D}^\prime(M)$ generalizing previous results
of Meyer and Brunetti--Fredenhagen. When
$X$ is a closed submanifold
of $M$,
we show that the concept of 
distributions with moderate growth
coincides with 
weakly homogeneous distributions of Meyer
which can be intrinsically defined. Then we renormalize products
of distributions with
functions tempered along $X$ and finally, 
using the whole analytical machinery developed,
we give an existence proof
of perturbative quantum field theories
on Riemannian manifolds.
\end{abstract}


\tableofcontents

\section*{Introduction.}
Let 
us start
with the following
example
which is
discussed
in \cite[Example 9 p.~140]{ReedSimonI}
and actually
goes back to
Hadamard. We denote by
$\Theta$ the Heaviside
function (the indicator function
of $\mathbb{R}_{\geqslant 0}$),
consider the
function
$x^{-1} \Theta(x)$
viewed as a distribution
in $\mathcal{D}^\prime(\mathbb{R}\setminus \{0\})$.
Obviously,
the linear
map
\begin{equation}
\varphi\longmapsto
\int_0^\infty dx\frac{\varphi(x)}{x} 
\end{equation}
is ill-defined if $\varphi(0)\neq 0$
since the integral $\int_0^\infty \frac{dx}{x}$ 
diverges.

However,
the integral
$\int_0^\infty dx x^{-1}\varphi(x)$
\textbf{converges}
if $\varphi(0)=0$ 
and an elementary estimate
shows that
$x^{-1}\Theta(x)$
defines a linear functional
on the ideal of functions
$x\mathcal{D}(\mathbb{R})$
vanishing at $0$.
A test function
$\varphi\in \mathcal{D}(\mathbb{R})$
being given, note that
the following
expression
\begin{equation}
\underset{\varepsilon\rightarrow 0}{\lim}\int_\varepsilon^1 dx 
\frac{(\varphi(x)-\varphi(0))}{x} +\int_1^\infty dx \frac{\varphi(x)}{x}
\end{equation}
converges.

We thus define
a \textbf{renormalized} distribution:
\begin{eqnarray}\label{fpxalpha}
x^{-1}_+=\lim_{\varepsilon\rightarrow 0}
\int_\varepsilon^\infty dxx^{-1} + \log(\varepsilon)\delta
\end{eqnarray}
where we subtracted 
the 
distribution 
$\log(\varepsilon)\delta$
supported at $0$, which becomes
singular when $\varepsilon\rightarrow 0$,
called
\emph{local counterterm}.
The renormalized distribution
$x^{-1}_+\in\mathcal{D}^\prime(\mathbb{R})$,
called finite part of Hadamard,
extends the
linear functional
$x^{-1}\Theta(x)\in\left(x\mathcal{D}(\mathbb{R})\right)^\prime$.
Our example shows the most elementary situation
where we can extend a distribution
by an \emph{additive renormalization}.

\textbf{In what follows, $M$ will always
denote a
smooth, paracompact and oriented manifold.}
In our paper, we investigate 
the following
problem which has simple formulation:
we are given a 
manifold $M$ and a closed subset
$X\subset M$. 
We 
define a natural
growth condition
on $t\in\mathcal{D}^\prime(M\setminus X)$
which measures the singular behaviour
near $X$ and we address the following
problems:
\begin{enumerate}
\item can we find a distribution
$\overline{t}\in\mathcal{D}^\prime(M)$ 
s.t. the restriction of $\overline{t}$ on $M\setminus X$ coincides with $t$,
\item can we construct a linear
extension operator $\mathcal{R}$, 
eventually give explicit formulas for $\mathcal{R}$,
\item can we classify the different extension
operators.
\end{enumerate}

In general, the extension
problem
has no positive answer for 
a generic
distribution $t$
in 
$\mathcal{D}^\prime(M\setminus X)$
unless $t$
has moderate growth
when we approach the singular 
subset $X$.
\paragraph{Distributions
having moderate growth along a closed subset $X\subset M$.}
If $P$ is a differential operator
with smooth coefficients on $M$, and $K\subset U$
a compact subset,
we denote by
$\Vert \varphi\Vert^K_P$ (resp $\Vert \varphi\Vert_P$)
the seminorm 
$\sup_{x\in K}\vert P\varphi(x)\vert$
(resp $\sup_{x\in U}\vert P\varphi(x)\vert$).
We also denote by $d$ some arbitrary 
distance function induced by some choice
of smooth metric on $M$. For every open set $V\subset M$, 
we denote by $\mathcal{T}_{M\setminus X}(V)$ the set
of distributions in $\mathcal{D}^\prime(V\setminus X)$ 
with moderate growth along $X$ defined as follows:
\begin{defi}
A distribution 
$t\in \mathcal{D}^\prime(V\setminus X)$ 
has moderate growth  along $X$
if for all open relatively compact $U\subset V$, 
there is a seminorm $\Vert .\Vert_P$ and a pair of constant
$(C,s)\in\mathbb{R}_{\geqslant 0}^2$ such that
\begin{eqnarray}\label{moderategrowthestimate}
\vert t(\varphi)\vert \leqslant C(1+d(\text{supp }\varphi,X)^{-s})\Vert \varphi\Vert_P.
\end{eqnarray}
for all $\varphi\in \mathcal{D}(U\setminus X)$.
\end{defi}

\textbf{Remark}: If $t$ were in $\mathcal{D}^\prime(M)$, we would
have the same estimate without the divergent
factor $(1+d(\text{supp }\varphi,X)^{-s})$.

The space $\mathcal{T}_{M\setminus X}$
is intrinsically defined
since all metrics on $M$
are locally equivalent.
The first main theorem we shall prove is
\begin{thm}\label{mainthm}
The three following claims
are equivalent:
\begin{enumerate}
\item $t$ has moderate growth along $X$,
\item $t\in\mathcal{D}^\prime(M\setminus X)$
is extendible, 
\item there is a family of functions
$(\beta_\lambda)_{\lambda\in(0,1]}\subset C^\infty(M\setminus X)$,
$\beta_\lambda=0$ in a neighborhood of $X$, $\beta_\lambda\underset{\lambda\rightarrow 0}{\rightarrow} 1$
and a family of distributions $(c_\lambda)_{\lambda\in(0,1]}$ \textbf{supported on }$X$
such that
\begin{equation}
\lim_{\lambda\rightarrow 0}t\beta_\lambda -c_\lambda
\end{equation}
exists and defines an extension of $t$ in $\mathcal{D}^\prime(M)$.
\end{enumerate}
\end{thm}
Our moderate growth condition is weaker than 
the hypothesis of \cite[Lemma 3.3]{KashiwaraRiemannHilbert}
and Theorem \ref{mainthm} can also be viewed as generalizations
of Theorem \cite[Thm 2.1 p.~48]{Meyer-98} and \cite[Thm 5.2 p.~645]{Brunetti2} which only 
treat
the extension problem
in the case of a point. 
When
$X$ is a vector subspace of $M=\mathbb{R}^n$,
we prove in Theorem \ref{scalingmodgrowth} that weakly
homogeneous distributions in the sense of Meyer
have moderate growth and are therefore extendible. 
In \cite[Chapter 1]{Dangthese}, we proved 
that weakly homogeneous distributions
along some vector subspace $X$ are \textbf{invariant by diffeomorphisms
preserving} $X$ which implies
that weakly homogeneous distributions along a submanifold $X\subset M$
can be intrinsically defined.

In the third part of our paper, we apply
our extension techniques
to establish in Theorem \ref{renormproduct} that the
product of distributions in $\mathcal{D}^\prime(M)$
with functions which are
tempered along $X$ (see definition \ref{tempdefi} for 
the algebra
$\mathcal{M}(X,M)$ of tempered functions) is renormalizable
which implies that the space of extendible
distributions or equivalently of distributions
in $\mathcal{T}_{M\setminus X}$ 
is a left $\mathcal{M}(X,M)$-module (Corollary \ref{tempfunmodules}). 

Finally we apply our analytic machinery
to the study of perturbative QFT on Riemannian manifolds.
In QFT, 
one is interested in making sense
of correlation functions 
denoted by
$\left\langle :\phi^{i_1}:(x_1)\dots :\phi^{i_n}:(x_n) \right\rangle$
which are objects living
in the configuration space $M^n$
that can be expressed formally, using the Feynman rules, 
in terms of
products of the form
$\underset{1\leqslant i<j\leqslant n}{\prod} G(x_i,x_j)^{n_{ij}} $
where $G$ is the Green function of $\Delta_g+m^2$ where $\Delta_g$ is the Laplace Beltrami operator.
A product $\underset{1\leqslant i<j\leqslant n}{\prod} G(x_i,x_j)^{n_{ij}} $
is called \emph{Feynman amplitude} and is depicted pictorially
by a graph with $n$ labelled vertices $\{1,\dots,n\}$ where
the vertices $i$ and $j$ are connected by $n_{ij}$ lines. 
In the second main Theorem (Thm \ref{renormthmqft}) of our paper, we prove that all Feynman amplitudes 
are renormalizable by a collection of extension
maps $(\mathcal{R}_n)_{n\in\mathbb{N}}$ where every map
$\mathcal{R}_n$ extends Feynman amplitudes living on the configuration
space $M^n$ minus all \emph{diagonals} to distributions on $M^n$ and the maps 
$(\mathcal{R}_n)_{n\in\mathbb{N}}$ satisfy some
axioms (definition \ref{axiomsrenormmaps}) which are due
to N. Nikolov \cite{NST}. This gives a different approach to
Costello's existence Theorem \cite{Costello,CostelloGwilliam} 
for perturbative QFT
on Riemannian manifolds.

\paragraph{Related works.}

In the litterature, the idea to consider
extendible distributions really goes back to 
Lojasiewicz \cite{Lojasiewicz} and tempered functions already appear
in the work of B. Malgrange~\cite{Malgrange,Malgrangeideals}. 
 However, the first general definition 
of a tempered distribution on
any open set $U$ in some manifold $M$ is due to M. Kashiwara, a distribution is tempered
if it is extendible
on $\overline{U}$ \cite[Lemma 3.2 p.~332]{KashiwaraRiemannHilbert} (see also \cite{CHM}) which implies by our Theorem  \ref{mainthm} that these distributions are in 
$\mathcal{T}_{M\setminus \partial U}$ i.e. have moderate growth
along $\partial U$.
His approach was then extended in \cite{GuillermouSchapira,KSmodcohom,KSindsheaves}.  
Tempered functions and distributions were also recently 
studied in the context of real algebraic geometry
\cite{AizenbudGourevitch, CHM} with applications in representation
theory.
More recently, a different approach to the extension problem
in terms of scaling was developped by Meyer in his book \cite{Meyer-98}, his purpose
was to study the singular behaviour at given points of irregular functions with applications
in multifractal analysis \cite{JaffardMeyer}. 
Our goal in this paper
is to revive some techniques in analysis originally developped by H. Whitney \cite{Whitney} which were
then improved by Malgrange and Lojasiewicz, to compare these techniques with the approach
by scaling of Meyer \cite{Dangthese,Meyer-98} and finally show their relevance 
in solving the problem of constructing a perturbative quantum field theory on a Riemannian manifold.  

\paragraph{Acknowledgements.}

I would like to thank Christian Brouder, Fr\'ed\'eric H\'elein,
Stefan De Bi\`evre, Laura Desideri, Camille Laurent Gengoux, Mathieu Sti\'enon
for useful discussions and the Labex CEMPI for excellent working conditions.

\section{The extension of distributions.}

\subsection{Proof of Theorem \ref{mainthm}.}\label{Partitionofunityargument}
\paragraph{Localization on open charts by a partition of unity.}

We shall reduce the proof of $(1)\Leftrightarrow (2)$
in Theorem \ref{mainthm}
to the case where $M=\mathbb{R}^n$, $X$ is a compact set
contained in a larger compact $K$ 
and $t\in\mathcal{D}^\prime(\mathbb{R}^n\setminus X)$
vanishes outside $K$, this condition reads
$t\in\mathcal{D}^\prime_K(\mathbb{R}^n\setminus X)$.
The first step is to localize the problem
by a partition of unity.
Choose a locally finite cover  of $M$
by relatively compact open charts $(U_i)_i$ 
and a subordinated 
partition of unity $(\varphi_i)_i$ s.t. $\sum \varphi_i=1$.
Denote by $t_i$ the restriction $t|_{U_i}$ and $K_i=\text{supp }\varphi_i\subset U_i$.
For all $\varphi\in\mathcal{D}(U)$, $t\in\mathcal{D}^\prime(U\setminus X)$ 
has moderate growth implies 
the same property for $t\varphi\in\mathcal{D}^\prime(U\setminus X)$, 
therefore each $t\varphi_i|_{U_i\setminus X}$ is in $\mathcal{D}_{K_i}^\prime(U_i\setminus (X\cap K_i))$, 
$t\varphi_i$ vanishes outside $K_i$ 
and has moderate growth along $X$.  
Hence
it suffices to extend $t\varphi_i|_{U_i\setminus X}$ in each $U_i$
in such a way that the extension is supported
by $K_i$.
Call
$\overline{t_i\varphi_i}$ such extension in
$\mathcal{E}^\prime(U_i)$
then the
locally finite sum 
$\overline{t}=\sum_i \overline{t_i\varphi_i}\in\mathcal{D}^\prime(M)$
is a well defined extension of $t$.
\paragraph{Working on $\mathbb{R}^n$.}
The second step is to use local charts to work on $\mathbb{R}^n$.
On every open set $(U_i)$,
let $\psi_i:U_i \longmapsto V\subset\mathbb{R}^n$ 
denote the corresponding chart
then the pushforward $\psi_{i*}(t\varphi_i)$ 
is in $\mathcal{D}_{\psi_i(K_i)}^\prime(V\setminus \psi_i(X\cap K_i))$.
Actually the compact set $\psi_i(X\cap K_i)$ is in the
interior of $V$, since
$(K_i\cap X)\subset  
\text{int}(U_i)$ and $\psi_i$ is a diffeomorphism.
Therefore
the distribution $\psi_{i*}(t\varphi_i)$ is an element
of $\mathcal{D}_{K_i}^\prime(\mathbb{R}^n\setminus \psi_i(X\cap K_i)) $
and
we may reduce the proof of our theorem 
to the case where we have a distribution
$t\in \mathcal{D}_K^\prime(\mathbb{R}^n\setminus X)$
with moderate growth along $X$
where $X\subset K$ are \textbf{compact subsets} of $\mathbb{R}^n$. 
In the sequel, we use the seminorms
$\Vert\varphi\Vert_m=\sup_{x\in \mathbb{R}^n,\vert\alpha\vert\leqslant m}\vert \partial_x^\alpha\varphi(x)\vert$
and
$\Vert\varphi\Vert_m^K=\sup_{x\in K ,\vert\alpha\vert\leqslant m}\vert \partial_x^\alpha\varphi(x)\vert$
where $K$ runs over the compact
subsets of $\mathbb{R}^n$.
Let
$\mathcal{I}(X,\mathbb{R}^n)=\{\varphi \text{ s.t. supp }\varphi\cap X=\emptyset \}\subset C^\infty(\mathbb{R}^n)$, 
since $t$ vanishes outside some compact set $K$,
the moderate growth condition
now reads
\begin{eqnarray}\label{moderategrowthcompactsupport}
\exists (C,s)\in\mathbb{R}_{\geqslant 0}^2\text{ and }\Vert .\Vert_{m}^K \text{ s.t. }  \forall \varphi\in \mathcal{I}(X,\mathbb{R}^n),
\vert t(\varphi)\vert\leqslant
C(1+d(\text{supp }\varphi,X)^{-s})\Vert\varphi\Vert^K_m.
\end{eqnarray}

\begin{thm}\label{euclideanmainthm}
Let $X\subset K$ be compact subsets of $\mathbb{R}^n$, 
then
$t\in \mathcal{D}_K^\prime(\mathbb{R}^n\setminus X)$ 
is extendible in $\mathcal{D}_K^\prime(\mathbb{R}^n)$
if and only if $t$ has moderate growth along $X$.
\end{thm}

\begin{proof}
We first
prove a weaker equivalence:
$t$ is extendible iff the estimate 
(\ref{moderategrowthcompactsupport}) holds with $s=0$.

Assume the problem is solved and that we could 
find an extension
$\overline{t}\in\mathcal{D}_K^\prime(\mathbb{R}^n)$ of $t$.
Observe that $\forall\varphi\in V, t(\varphi)=\overline{t}(\varphi)$
then by definition 
$\overline{t}$ is a linear continuous functional
on $C^\infty(\mathbb{R}^n)$ equipped with the Fr\'echet topology, thus it induces
a linear continuous map
on the vector subspace
$\mathcal{I}(X,\mathbb{R}^n)\subset C^\infty(\mathbb{R}^n)$:
\begin{eqnarray*} 
\exists C \in\mathbb{R}_{\geqslant 0}, \Vert .\Vert_{m}^K \text{ s.t. }  \forall \varphi\in \mathcal{I}(X,\mathbb{R}^n),\,\
\vert t(\varphi)\vert=\vert \overline{t}(\varphi)\vert\leqslant
C\Vert\varphi\Vert_m^K.
\end{eqnarray*}
Therefore, if $t$ is extendible then
estimate (\ref{moderategrowthcompactsupport}) is satisfied
with $s=0$ and $t$ has moderate growth along $X$.

Conversely, if 
$\exists C \in\mathbb{R}_{\geqslant 0}, \Vert .\Vert_{m}^K \text{ s.t. }
 \forall \varphi\in \mathcal{I}(X,\mathbb{R}^n),
\vert t(\varphi)\vert\leqslant C
\Vert\varphi\Vert^K_m$, then by the Hahn--Banach
theorem \cite[Thm 6.4 p.~46]{VogtMeise}, we can extend $t$ as a linear continuous
mapping $\overline{t}$ on $C^\infty(\mathbb{R}^n)$ which satisfies the above estimate
hence $\overline{t}\in \mathcal{D}_K^\prime(\mathbb{R}^n)$.
Therefore to prove that $t$ has \emph{moderate growth} implies 
that $t$ is extendible in $\mathcal{D}^\prime_K(\mathbb{R}^n)$, 
it suffices to show that
\begin{eqnarray*}
&&\exists C \in\mathbb{R}_{\geqslant 0}, \Vert .\Vert_{m}^K \text{ s.t. }\forall \varphi\in \mathcal{I}(X,\mathbb{R}^n),
\vert t(\varphi)\vert\leqslant
C(1+d(\text{supp }\varphi,X)^{-s})\Vert\varphi\Vert_m^K\\
 &\implies & \exists C^\prime \in\mathbb{R}_{\geqslant 0}, \Vert .\Vert_{m^\prime}^{K} \text{ s.t. }
\forall \varphi\in \mathcal{I}(X,\mathbb{R}^n),\vert t(\varphi)\vert\leqslant
C^\prime\Vert\varphi\Vert_{m^\prime}^{K}.
\end{eqnarray*}

Let us admit the following central technical Lemma 
whose proof will be given later:
\begin{lemm}\label{Malgrangestyletechnicallemma}
For every integer $d\in\mathbb{N}$, let 
$\mathcal{I}^{m+d}(X,\mathbb{R}^n)$ denote
the closed ideal
of functions of regularity $C^{m+d}$ 
which vanish at order $m+d$ on $X$.
Then there is a function $\chi_\lambda\in C^\infty(\mathbb{R}^n)$
parametrized by $\lambda\in(0,1]$
s.t. $\chi_\lambda=1$ (resp $\chi_\lambda=0$) when $d(x,X)\leqslant\frac{\lambda}{8}$
(resp $d(x,X)\geqslant\lambda$)
\begin{eqnarray}
\exists \tilde{C}, \forall\lambda\in(0,1], \forall \varphi \in \mathcal{I}^{m+d}(X,\mathbb{R}^n),
\Vert \chi_\lambda \varphi\Vert^K_m\leqslant \tilde{C}\lambda^d  \Vert \varphi\Vert^{K\cap\{d(x,X)\leqslant \lambda\}}_{m+d} 
\end{eqnarray} 
where the constant $\tilde{C}$ does not depend on $\varphi,\lambda$.
\end{lemm}

If $s=0$, then we know that 
there is an extension by Hahn Banach therefore
we shall treat the case where $s>0$. Our idea	
is to absorb the divergence
by a dyadic decomposition:
\begin{eqnarray*}
&&\forall \varphi\in \mathcal{I}(X,\mathbb{R}^n),\exists N
\text{ s.t. } \chi_{2^{-N}}\varphi=0\\
&\implies & t(\varphi)=t((1-\chi_{2^{-N}})\varphi)\\
&\implies &t(\varphi)=\sum_{j=0}^{N-1}t((\chi_{2^{-j}}-\chi_{2^{-j-1}})\varphi)+t((1-\chi_{1})\varphi )
\end{eqnarray*}
We easily estimate $t((1-\chi_{1})\varphi )$:
$\forall\varphi\in C^\infty(\mathbb{R}^n),\vert t((1-\chi_{1})\varphi )\vert\leqslant C\Vert\varphi\Vert_m$
for some constant $C$ since
the support  of $1-\chi_{1}$ does not meet $X$. 
Choose $d\in\mathbb{N}^*$ such that
$d-s>0$, then:
\begin{eqnarray*}
\vert t(\chi_{1}\varphi ) \vert
&\leqslant &
\sum_{j=0}^{N-1}\vert t((\chi_{2^{-j}}-\chi_{2^{-j-1}})\varphi)\vert\\
&\leqslant & C\sum_{j=1}^{N} (1+d(\text{supp }\varphi(\chi_{2^{-j}}-\chi_{2^{-j-1}}),X)^{-s})
\Vert (\chi_{2^{-j}}-\chi_{2^{-j-1}})\varphi \Vert_m^K\text{, by moderate growth } \\
&\leqslant & C\sum_{j=1}^{N} (1+2^{s(j+4)})(2^{-jd}+2^{-(j+1)d})\tilde{C}\Vert\varphi \Vert^K_{m+d}\text{, by the technical Lemma } \\ 
&\leqslant & C^\prime \Vert\varphi \Vert^K_{m+d}
\end{eqnarray*}
for $C^\prime=\tilde{C}C(1+2^{-d})\underset{\text{convergent series since }d-s>0}{\underbrace{\sum_{j=1}^{\infty} 2^{-jd}(1+2^{(j+4)s})}} <+\infty$
which is independent of $N$ and $\varphi$.  
\end{proof}
We now prove Lemma \ref{Malgrangestyletechnicallemma}:
\begin{proof}
Choose $\phi\geqslant 0$ s.t. $\int\phi=1$, 
$\phi=0$ if $\vert x\vert\geqslant \frac{3}{8}$
then set $\phi_\lambda=\lambda^{-n}\phi(\lambda^{-1}.)$ 
and 
$\alpha_\lambda$ to be the characteristic function
of the set $\{x \text{ s.t. }d(x,X)\leqslant \frac{\lambda}{2} \}$ 
then the convolution
product $\phi_\lambda*\alpha_\lambda(x)=1$ 
if $d(x,X)\leqslant \frac{\lambda}{8}$ and equals $0$
if $d(x,X)\geqslant \lambda$.
Since by Leibniz rule
one has $\partial^\alpha (\chi_\lambda\varphi)(x)=\underset{\vert k\vert\leqslant \vert\alpha\vert}\sum\left(
\begin{array}{c}
\alpha\\
k
\end{array}
\right) \partial^k\chi_\lambda\partial^{\alpha-k}\varphi(x) $,
it suffices to estimate
each term
$\partial^k\chi_\lambda\partial^{\alpha-k}\varphi(x)$
of the above sum.

For all multi-index $k$, there is some constant
$C_k$ such that
$\forall x\in \mathbb{R}^n\setminus X, \vert \partial_x^k 
\chi_\lambda \vert\leqslant
\frac{C_k}{\lambda^{\vert k\vert}}$
and $\text{supp }\partial_x^k 
\chi_\lambda\subset\{d(x,X)\leqslant \lambda \}$. 
Therefore for all $\varphi\in\mathcal{I}^{m+d}(X,\mathbb{R}^n)$, 
for all
$x\in \text{supp }\partial_x^k \chi_\lambda\partial^{\alpha-k}\varphi$,
for $y\in X$ such that
$d(x,X)=\vert x-y\vert$, we find that $\partial^{\alpha-k}\varphi$
vanishes at $y$ at order $\vert k\vert+d $ therefore:
$$\partial_x^{\alpha-k}\varphi(x) =\sum_{\vert\beta\vert=\vert k\vert+d}
(x-y)^\beta R_\beta(x)$$
where the right hand side
is just the integral remainder
in Taylor's expansion of $\partial^{\alpha-k}\varphi$ around $y$.
Hence:
$$ \vert\partial^k\chi_\lambda\partial^{\alpha-k}\varphi(x)   \vert
\leqslant \frac{C_k}{\lambda^{\vert k\vert}}\sum_{\vert\beta\vert=\vert k\vert+d}
\vert (x-y)^\beta R_\beta(x)\vert  .$$
 It is easy to see 
that $R_\beta$ only
depends on the Jets 
of $\varphi$ of order $\leqslant m+d$.
Hence
$$\vert\partial^k\chi_\lambda\partial^{\alpha-k}\varphi(x)   \vert
\leqslant 
C_{k} \lambda^d \sup_{x\in K, d(x,X)\leqslant \lambda}\sum_{\vert\beta\vert=\vert k\vert+d}
\vert R_\beta(x)\vert$$
and the conclusion follows easily.
\end{proof}

Our partition of unity argument together with the
result of Theorem \ref{euclideanmainthm}
imply that $(1)\Leftrightarrow (2)$
in Theorem \ref{mainthm}.
\subsection{Renormalizations and the Whitney extension Theorem.}

The goal of this subsection is to replace the use 
of Hahn Banach theorem by a more constructive argument.
First, we discuss a particular case of extension
where there is some canonical choice for $\overline{t}$.

\paragraph{Remark on the extension of positive measures with locally finite mass.}

The following
proposition is inspired by some results of Skoda \cite{Skoda}. 
Let $\mu$ be a positive measure in $M\setminus X$,
then we say that $\mu$ has locally finite mass if:
\begin{eqnarray*}
\forall K\subset M \text{ compact },\exists C_K,
\forall\varphi\in \mathcal{D}_K(M\setminus X),
\varphi\geqslant 0,
0\leqslant\mu(\varphi) \leqslant C_K\Vert\varphi\Vert_0. 
\end{eqnarray*}
\begin{prop}
Let $\mu$ be a positive measure in $M\setminus X$.
If
$\mu$ has locally finite mass then $\mu$ has a canonical extension in the space of positive measures.
\end{prop}
\begin{proof}
By an obvious regularization argument, we can extend
$\mu$ to the space $C^0_c(M\setminus X)$
of compactly supported functions of regularity $C^0$.
Choose a family $\chi_\lambda$ as in 
the main technical Lemma \ref{Malgrangestyletechnicallemma} which
satisfies $\chi_\lambda\geqslant 0, \chi_\lambda=1$ if $d(x,X)\leqslant \frac{\lambda}{8}$
and $\chi_\lambda=0$ when $d(x,X)\geqslant\lambda$.
Then for all $\varphi\in C^0_c(M), \varphi\geqslant 0$, the sequence 
$\mu((1-\chi_{2^{-n}})\varphi)_n$ is increasing and bounded by 
$C_K\Vert\varphi \Vert_0$ where $K$ is any compact set
which contains the support of $\varphi$. 
Therefore for each $\varphi\geqslant 0$,
$\lim_{n\rightarrow +\infty}\mu((1-\chi_{2^{-n}})\varphi)$
exists. It is easy to conclude using the fact
that $C^0_c(M)$ is spanned by non negative
functions. 
\end{proof}

\paragraph{Constructive extension operator instead of Hahn Banach.}

Recall we denote by $\mathcal{I}(X,\mathbb{R}^n)$
the smooth functions vanishing in some neighborhood of $X$.
In the proof of Theorem \ref{euclideanmainthm}, we
showed that if $t$ were extendible equivalently 
if $t$ satisfies the moderate growth condition then:
\begin{eqnarray}\label{tcontinuitycm}
 \exists (C,m) , \forall\varphi\in \mathcal{I}(X,\mathbb{R}^n) ,
\vert t(\varphi) \vert \leqslant C \Vert\varphi\Vert^K_{m}.
\end{eqnarray}
Therefore $t$ defines a linear functional on 
$\mathcal{I}(X,\mathbb{R}^n)$
for the induced topology of $C^\infty(\mathbb{R}^n)$ and can
be extended by
Hahn Banach which
is a non constructive argument
and \textbf{does not imply} the existence of a linear extension operator
$t\in\mathcal{D}_K^\prime(\mathbb{R}^n\setminus X)\longmapsto \overline{t}\in\mathcal{D}_K^\prime(\mathbb{R}^n)$.

Denote by
$\mathcal{I}^{m}(X,\mathbb{R}^n)$
the space of $C^{m}$ functions
which vanish on $X$ together with all their derivatives
of order less than $m$, $\mathcal{I}^{m}(X,\mathbb{R}^n)$ 
is a closed ideal 
in $C^m(\mathbb{R}^n)$. 
To construct a linear extension operator, we have
to prove first that $t$ extends by continuity to some element 
$t_m$
in the topological dual $\mathcal{I}^{m}(X,\mathbb{R}^n)^\prime$ 
of $\mathcal{I}^{m}(X,\mathbb{R}^n)\subset
C^m(\mathbb{R}^n)$.
\begin{lemm}
$t$ satisfies (\ref{tcontinuitycm}) if and only if
$t$ uniquely extends by continuity to an element
$t_m$ in $\mathcal{I}^{m}(X,\mathbb{R}^n)^\prime$:
\begin{eqnarray}
\forall\varphi\in \mathcal{I}^{m}(X,\mathbb{R}^n), t_m(\varphi)
&=&
\lim_{\lambda\rightarrow 0}\lim_{\varepsilon\rightarrow 0}t((1-\chi_\lambda)\phi_\varepsilon*\varphi)
\end{eqnarray}
for the family of cut--off functions 
$(\chi_\lambda)_\lambda$ defined in Lemma \ref{Malgrangestyletechnicallemma}
and a mollifier $\phi_\varepsilon$.
\end{lemm}
\begin{proof}
It suffices to prove that
the space of $C^{\infty}$ functions
whose support does not meet $X$ is dense in 
$\mathcal{I}^{m}(X,\mathbb{R}^n)$ in the $C^m$ topology.
In fact, we prove more, let $\phi_\varepsilon$ be a  
smooth mollifier, then by a classical regularization argument,  we have
$\lim_{\varepsilon\rightarrow 0}
(1-\chi_\lambda)\phi_\varepsilon*\varphi=(1-\chi_\lambda)\varphi$ in $C^m(\mathbb{R}^n)$
for all $\varphi\in C^m(\mathbb{R}^n)$
and $\lim_{\lambda\rightarrow 0} (1-\chi_\lambda)\varphi\rightarrow \varphi$ 
in $\mathcal{I}^m(X,\mathbb{R}^n)$.  
By the technical Lemma \ref{Malgrangestyletechnicallemma} (see \cite{Malgrangeideals} p.~11), we have
$$\forall\varphi \in \mathcal{I}^{m}(X,\mathbb{R}^n), \Vert\chi_\lambda\varphi\Vert^K_m\leqslant \tilde{C}\Vert \varphi\Vert_m^{K\cap \{d(x,X)\leqslant \lambda\}}\rightarrow 0$$ when $\lambda\rightarrow 0$ therefore
$\varphi=\lim_{\lambda\rightarrow 0}(1-\chi_\lambda)\varphi$
in the $C^m$ topology. Finally this proves 
$\mathcal{I}^{m}(X,\mathbb{R}^n)$ is the closure
in $C^m(\mathbb{R}^n)$ of
the space of $C^{\infty}$ functions
whose support does not meet $X$.
\end{proof} 
 Set $\beta_\lambda=1-\chi_\lambda$, from the above Theorem
we can make a notation abuse and say that
$\underset{\lambda\rightarrow 0}{\lim}t\beta_\lambda\in \mathcal{I}^{m}(X,\mathbb{R}^n)^\prime$ if $t$ satisfies
the estimate (\ref{tcontinuitycm}) 
(we just forget about the mollifier).
The idea is to compose $\underset{\lambda\rightarrow 0}{\lim}t\beta_\lambda$ with a continuous
projection $I_m:C^m(\mathbb{R}^n)\longmapsto \mathcal{I}^m(X,\mathbb{R}^n)$ so
that $\underset{\lambda\rightarrow 0}{\lim}t\beta_\lambda\circ I_m$ defines an extension
of $t$.
Dually, every compactly supported
distribution
of order $m$ induces by restriction
a linear functional on $\mathcal{I}^m(X,\mathbb{R}^n)$ , in other words
we have a surjective linear map
$p:\mathcal{E}^\prime_m(\mathbb{R}^n) \mapsto \mathcal{I}^m(X,\mathbb{R}^n)^\prime$.
We want to construct  a 
linear extension operator 
$\mathcal{R}$ called a \emph{renormalization map}
from $\mathcal{I}^{m}(X,\mathbb{R}^n)^\prime$ to $\mathcal{E}_m^\prime(\mathbb{R}^n)$
such that $p\circ \mathcal{R}:\mathcal{I}^{m}(X,\mathbb{R}^n)^\prime\mapsto \mathcal{I}^{m}(X,\mathbb{R}^n)^\prime$ is the identity map.
Then it is immediate to note that the transpose of $\mathcal{R}$
is the projection $I_m$.

Denote by $\mathcal{E}^m(X)$ the
space of differentiable functions of order
$m$ in the sense of Whitney \cite[Definition 2.3 p.~3]{Malgrangeideals},\cite[p.~146]{Bierstone}.
\begin{thm}\label{Whitneythm}
There is a bijection
between:
\begin{itemize}
\item the space of 
\emph{renormalization maps}
\item the space of decompositions
of $C^m(\mathbb{R}^n)$ in direct sum 
$C^m(\mathbb{R}^n)=\mathcal{I}^{m}(X,\mathbb{R}^n)\oplus B$
where $B$ is a closed subspace
of $C^m$ which we call
\emph{renormalization scheme}
\item the space of continuous linear  
splittings
of the exact sequence
\begin{eqnarray}\label{Whitneyexactseq}
0\longmapsto \mathcal{I}^{m}(X,\mathbb{R}^n)  \longmapsto \mathcal{C}^m(\mathbb{R}^n)
\overset{q}{\rightarrow} \mathcal{E}^m(X)\longmapsto 0. 
\end{eqnarray}
\end{itemize}
\end{thm}

\begin{proof}
The exactness of (\ref{Whitneyexactseq}) and the existence of linear continuous splittings
of (\ref{Whitneyexactseq}) is a consequence of the Whitney extension theorem (see \cite[p.~10]{Malgrangeideals},
\cite[Thm 2.3 p.~146]{Bierstone}).
Since (\ref{Whitneyexactseq}) is a continuous exact sequence 
of Fr\'echet spaces,
the dual sequence:
\begin{eqnarray}\label{dualwhitneyseq}
0\longmapsto \mathcal{E}^\prime_{m,X}(\mathbb{R}^n) \longmapsto \mathcal{E}^\prime_{m}(\mathbb{R}^n)
\overset{p}{\rightarrow} \mathcal{I}^{m}(X,\mathbb{R}^n)^\prime\longmapsto 0 
\end{eqnarray}
is exact \cite[Prop 26.4 p.~308]{VogtMeise}.

$T$ is a linear splitting of (\ref{Whitneyexactseq}) 
\begin{itemize}
\item $\Leftrightarrow  T\circ q$ is a
continuous projector
on the closed subspace $B=\text{ran}(T)$
\item $\Leftrightarrow C^m(\mathbb{R}^n)=B\oplus \mathcal{I}^m(X,\mathbb{R}^n)$
where the projection $Id-T\circ q$ on $\mathcal{I}^m(X,\mathbb{R}^n)$ is denoted
by $I_m$
\item $\Leftrightarrow \mathcal{R}= ^tI_m$ splits
the dual exact sequence (\ref{dualwhitneyseq}). 
\end{itemize} 
\end{proof}
\paragraph{The Whitney extension Theorem, formal neighborhoods and extendible distributions.}
Let us
give several interpretations of 
the result of Theorem \ref{Whitneythm}.
First, the reader 
can think of 
the direct sum decomposition 
as a way to decompose a $C^m$ 
function as a sum of a 
``Taylor
remainder'' which vanishes at order $m$ on $X$ and 
a ``Taylor polynomial'' in $B$. If $X$ were a point, $\mathcal{E}^m(X)$
is isomorphic to the space $\mathbb{R}_m[X_1,...,X_n]$ of polynomials of degree $m$ in $n$ variables,
we can choose $B=\mathbb{R}_m[x_1,...,x_n]$
and the decomposition $B+\mathcal{I}^m$ is given by Taylor's 
formula.
For $\varphi\in C^m(\mathbb{R}^n)$, one can think
of $q(\varphi)\in \mathcal{E}^m(X)\simeq C^m(\mathbb{R}^n)/\mathcal{I}^m(X,\mathbb{R}^n)$
as the restriction of $\varphi$ to the \emph{infinitesimal neighborhood of $X$ of order $m$}.
More generally, let $\mathcal{I}^\infty(X,\mathbb{R}^n)$ be the closed ideal
of functions in $C^\infty(\mathbb{R}^n)$ which vanish at infinite order on $X$, this 
is a nuclear Fr\'echet space since it is a closed subspace of the nuclear
Fr\'echet space $C^\infty(\mathbb{R}^n)$. 
We can think of the space $\mathcal{E}(X)$ of $C^\infty$
functions in the sense of Whitney as some sort of 
$\infty$-jets in ``the transverse directions'' to $X$
since by the Whitney extension theorem, 
we have a continuous exact sequence of nuclear Fr\'echet 
spaces:
\begin{equation}
0\longmapsto \mathcal{I}^\infty(X,\mathbb{R}^n) \longmapsto C^\infty(\mathbb{R}^n) \longmapsto \mathcal{E}(X)\longmapsto 0
\end{equation}
which implies that $\mathcal{E}(X)$ is the quotient space $C^\infty(\mathbb{R}^n)/\mathcal{I}^\infty(X,\mathbb{R}^n)$.
When $X$ is a submanifold of $\mathbb{R}^n$, it is
interesting to think of $\mathcal{E}(X)$ as smooth functions
\emph{restricted to the formal neighborhood of $X$}.
And the \emph{formal neighborhood of} $X$ is then 
defined
as the topological 
dual of $\mathcal{E}(X)$ 
which is nothing but the
space of distributions $\mathcal{E}^\prime_X(\mathbb{R}^n)$ 
with compact support contained in $X$ and fits in 
the continuous dual exact sequence of DNF spaces \cite[appendix A]{CHM}:
\begin{equation}
0\longmapsto \mathcal{E}^\prime_X(\mathbb{R}^n) \longmapsto \mathcal{E}^\prime(\mathbb{R}^n) \longmapsto \mathcal{E}(X)/\mathcal{E}^\prime_X(\mathbb{R}^n) \longmapsto 0
\end{equation}
where the quotient 
space $\mathcal{E}(X)/\mathcal{E}^\prime_X(\mathbb{R}^n)$ should be interpreted 
as the space of distributions in $\mathcal{D}^\prime(\mathbb{R}^n\setminus X)$ 
which are \textbf{extendible} in $\mathcal{E}(X)$ and the continuous map 
$\mathcal{E}^\prime(\mathbb{R}^n) \longmapsto \mathcal{E}(X)/\mathcal{E}^\prime_X(\mathbb{R}^n)$ is 
in fact the transpose of the inclusion map $\mathbb{R}^n\setminus X\hookrightarrow \mathbb{R}^n$. 
Another nice consequence
of the theory of nuclear Fr\'echet spaces is that the space
of extendible distributions is a DNF space.
 
\paragraph{The renormalization group.}
We also define the renormalization group $G$ as
the collection of linear,
continuous, bijective maps from $C^m(\mathbb{R}^n)$ to itself
preserving $\mathcal{I}^m(X,\mathbb{R}^n)$.
Note that $g\in G\implies g^{-1}$
is continuous by the open mapping
theorem hence $G$ is well defined as 
a group. 
Let $\mathcal{R}$ be a renormalization map 
corresponding to a projection $I_m$.
For any element $g\in G$, we define the action of $g$ on $\mathcal{R}$ as follows:
$\forall t\in\mathcal{I}^m(X,\mathbb{R}^n)^\prime, g.\mathcal{R}t(\varphi)=\mathcal{R}t(g(\varphi))=
t(I_m\circ g(\varphi))$ where 
$\mathcal{R}t(g .)\in\mathcal{E}^\prime(\mathbb{R}^n)$ is an extension of $t\in\mathcal{I}^m(X,\mathbb{R}^n)^\prime$
since $g$ preserves $\mathcal{I}^m(X,\mathbb{R}^n)$. 


\paragraph{Renormalization as subtraction of counterterms.}

Assume we choose a renormalization scheme.
We denote by $P_m=Id-I_m$ the projection from
$C^m$ to the closed subspace $B\subset C^m$ 
which plays the role
of the \emph{Taylor polynomials}.
From the above theorem and recall $\beta_\lambda=1-\chi_\lambda$
where $\chi_\lambda$ is the function of Lemma \ref{Malgrangestyletechnicallemma} 
\begin{prop}\label{counterterm}
If $t$ satisfies the estimate
\ref{tcontinuitycm} then: 
\begin{eqnarray}
\forall\varphi\in C^\infty(\mathbb{R}^n), 
\overline{t}(\varphi)=\lim_{\lambda\rightarrow 0} \underset{\text{finite part}}{t(\beta_\lambda I_m\varphi)}
&=& \lim_{\lambda\rightarrow 0}t(\beta_\lambda\varphi)  -\underset{\text{singular part}} {t(\beta_\lambda P_m\varphi)}
\end{eqnarray}
is a well defined extension of $t$.
\end{prop}
We call such extension 
a \textbf{renormalization}.
The divergences of 
$t(\beta_\lambda \varphi )$ come from the fact that
$\varphi\notin \mathcal{I}^m(X,\mathbb{R}^n)$, however these divergences
are local in the sense they can be subtracted by 
the counterterm $t(\beta_\lambda P_m\varphi)$ which becomes singular when
$\lambda\rightarrow 0$ and only depends on \textbf{the restriction to $X$
of the $m$-jets of $\varphi$}
(since $\varphi$ vanishes near $X$ implies that $\varphi\in\mathcal{I}^m\implies 
P_m\varphi=0)$.
By construction, the renormalization group $G$ acts
\textbf{on the space of all renormalizations} of $t$.

\subsection{Going back to the
manifold case.}
\paragraph{Difference between two extensions.}
Following the notations of \ref{Partitionofunityargument}, 
recall that $(U_i)_i$ was our locally
finite open cover of $M$ by relatively compact sets.
On each open set $U_i$, we defined
a chart $\psi_i:U_i\mapsto V\subset\mathbb{R}^n$
and we considered a partition of unity 
$(\varphi_i)_i$ subordinated to $(U_i)_i$.
Let $t\in\mathcal{D}^\prime(M\setminus X)$ be a distribution
with moderate growth, then by Theorem \ref{euclideanmainthm} we may assume
that:
\begin{equation}
\forall U_i, \exists m_i\in\mathbb{N} , \exists C_i>0, 
\forall\varphi\in C^\infty(\mathbb{R}^n\setminus X\cap \text{supp }\varphi_i),
\vert \psi_{i*}(t\varphi_i)(\varphi)\vert\leqslant C_i\Vert \varphi \Vert_{m_i}.
\end{equation}
By Theorem \ref{euclideanmainthm}, we may find
an extension  $\overline{t}=\sum_i \overline{t\varphi_i} \in \mathcal{D}^\prime(M)$
in such a way that for every $i$,
$\overline{t\varphi_i}|_{U_i}$
has order $m_i$.
If we prescribe the \textbf{order of the extensions} on every $U_i$
to be equal to $m_i\in\mathbb{N}$, 
then 
two extensions $t_1,t_2$ will differ
on each $U_i$ by a distribution $t_1-t_2|_{U_i}$ of order $m_i$
supported on $X\cap U_i$.
\paragraph{How to renormalize in the
manifold case ?}

On each chart $\psi_i:U_i\mapsto V\subset \mathbb{R}^n$, 
we can extend  
$\psi_{i*}(t\varphi_i)\in\mathcal{D}^\prime(V\setminus \psi_i(X\cap \text{ supp }\varphi_i))$
by renormalization. In other words, by Proposition
\ref{counterterm},
there is a family
of functions $\beta_{\lambda}(i)\in C^\infty(\mathbb{R}^n),\beta_\lambda(i)\rightarrow 1$ and counterterms
$c_\lambda(i)\in\mathcal{E}_{\psi_i(X\cap \text{ supp }\varphi_i)}^\prime(\mathbb{R}^n)$
such that $\lim_{\lambda\rightarrow 0} \psi_{i*}(t\varphi_i)\beta_\lambda(i)-c_\lambda(i)$ is an extension
of $\psi_{i*}(t\varphi_i)$ in $\mathcal{E}^\prime(\mathbb{R}^n)$.
Then setting 
\begin{eqnarray}\label{renormonchart}
\beta_\lambda=\sum_{i}\varphi_i\psi_i^*\beta_\lambda(i)
\text{ and }c_\lambda=\sum_{i}\psi_i^*c_\lambda(i),
\end{eqnarray}
we find that:
\begin{eqnarray}\label{renormchart2}
t\beta_\lambda-c_\lambda=\sum_{i}t\varphi_i\psi_i^*\beta_\lambda(i)-\psi_i^*c_\lambda(i) 
\end{eqnarray}
converges to some extension 
of $t$ when $\lambda\rightarrow 0$. This proves 
$(1)\Leftrightarrow (3)$ in Theorem (\ref{mainthm}).

\section{Moderate growth and scaling.}
In this section, we compare two approaches that
were developped to measure the singular behaviour of a distribution
along a closed subset $X$: the moderate growth condition
and the one used in \cite{Dangthese,Meyer-98,Brunetti2} in terms
of scaling. We show that both 
approach are equivalent when $X$ is a submanifold of $M$.
\subsection{Weakly homogeneous distributions have moderate growth.}
In this subsection,  
we work on $\mathbb{R}^n$ viewed as a product 
$\mathbb{R}^{n_1}\times \mathbb{R}^{n_2}, n=n_1+n_2$
and we adopt
the following splitting of variables 
$x\in \mathbb{R}^n=(x_1,x_2)\in \mathbb{R}^{n_1}\times \mathbb{R}^{n_2}$.
Here we establish the relationship between
our definition of moderate growth and the one
used by Yves Meyer \cite{Meyer-98} and 
the author \cite{Dangthese} in terms of scaling.
First we scale in the transverse directions
to a vector subspace $X= \mathbb{R}^{n_1}\times \{x_2=0\}$
of 
$\mathbb{R}^n$ with the maps
$\Phi^\lambda:(x_1,x_2)\longmapsto (x_1,\lambda x_2)$.
By definition, the scalings acts on $\mathcal{D}^\prime(\mathbb{R}^n)$
by duality $\left(\Phi^{\lambda*} t \right)(\varphi)=\lambda^{-n_2}t(\Phi^{\lambda^{-1}*}\varphi)$.
A distribution $t\in\mathcal{D}^\prime(\mathbb{R}^n\setminus X)$ 
is said to be weakly homogeneous in $\mathcal{D}^\prime(\mathbb{R}^n\setminus X)$
of degree $s$ if the family of distributions
$\lambda^{-s}\Phi^{\lambda*}t,\lambda\in(0,+\infty]$ is bounded in $\mathcal{D}^\prime(\mathbb{R}^n\setminus X)$.
\begin{thm}\label{scalingmodgrowth}
If $t$ is weakly homogeneous of degree $s$ in $\mathcal{D}^\prime(\mathbb{R}^n\setminus X)$
then $t$ has moderate growth along $X= \mathbb{R}^{n_1}\times \{x_2=0\}$.
More precisely, for all compact subset $K\subset\mathbb{R}^n$ there is $(m,C)\in\mathbb{N}\times\mathbb{R}$
and a compact subset $B\subset\mathbb{R}^n$ containing $K$
s.t.
\begin{eqnarray}
\forall\varphi\in \mathcal{D}_K(\mathbb{R}^n\setminus X), \vert t(\varphi)\vert\leqslant (1+d(\text{supp }\varphi,X)^{s+n_2})\Vert\varphi\Vert_m^{B}.
\end{eqnarray}
\end{thm}
It follows by Theorem \ref{mainthm} that
such $t$ has an extension in $\mathcal{D}^\prime(\mathbb{R}^n)$.
Note that when $s+n_2>0$, we are in a trivial situation of moderate growth
since the r.h.s. does not diverge.
\begin{proof}
The proof relies on the existence of 
a continuous partition of unity,
$$\int_0^\infty \frac{d\lambda}{\lambda} \psi(\lambda^{-1}x_2)=\int_0^\infty \frac{d\lambda}{\lambda} \Phi^{\lambda^{-1}*}\psi=1$$ where $\psi(\lambda^{-1}x_2)$
is supported on the corona $\frac{\lambda}{2}\leqslant \vert x_2 \vert \leqslant 2\lambda$.
Indeed, let $\chi\in C^\infty(\mathbb{R}^{n_2})$ 
be a function s.t. $\chi=1$ (resp $\chi=0$) 
when $\vert x\vert\leqslant\frac{1}{2}$ (resp $\vert x\vert\geqslant 2$) 
then set $\psi=-x\frac{d\chi}{dx}$.\\
Fix
a compact set
$B=\{\sup_{i=1,2}\vert x_i\vert\leqslant L\}$,
then for all test function
$\varphi\in\mathcal{D}_B(\mathbb{R}^n\setminus X)$
we obviously have
$$\varphi=\int_\varepsilon^{2L} \frac{d\lambda}{\lambda}\left(\Phi^{\lambda^{-1}*}\psi\right)\varphi\text{ for }\varepsilon \leqslant \frac{d(\text{supp }\varphi,X)}{2}, $$ since
$\lambda\notin [\frac{d(\text{supp }\varphi,X)}{2},2L] \implies \text{supp }\left(\Phi^{\lambda^{-1}*}\psi\right)\cap\text{supp }\left(\varphi\right)=\emptyset$.
Now it is obvious 
that
\begin{eqnarray*}
t(\varphi)&=& \int_{\frac{d(\text{supp }\varphi,X)}{2}}^{2L} \frac{d\lambda}{\lambda}t\left(\left(\Phi^{\lambda^{-1}*}\psi\right)\varphi \right)\\
&=&\int_{\frac{d(\text{supp }\varphi,X)}{2}}^{2L} \frac{d\lambda}{\lambda}\lambda^{s+n_2} \left(\lambda^{-s}\Phi^{\lambda*}t\right)\left(\psi\Phi^{\lambda*}\varphi\right)\\
\implies \vert t(\varphi)\vert&\leqslant & ((2L)^{s+n_2}+\left(\frac{d(\text{supp }\varphi,X)}{2}\right)^{s+n_2})\sup_{\lambda\leqslant 2L} \vert \left(\lambda^{-s}\Phi^{\lambda*}t\right)\left(\psi\Phi^{\lambda*}\varphi\right)  \vert
\end{eqnarray*}
A simple calculation
proves that $\left(\psi\Phi^{\lambda*}\varphi\right)_{\lambda\leqslant 2L}\subset \mathcal{D}_{\tilde{K}}(\mathbb{R}^n\setminus X)$
for $\tilde{K}=\{(x_1,x_2) | \vert x_1\vert\leqslant L, \frac{1}{2}\leqslant\vert x_2\vert\leqslant 2 \}$ and that:
\begin{eqnarray*}
\forall m\in\mathbb{N},\exists C_m>0,\forall\lambda, 
\Vert \psi\Phi^{\lambda*}\varphi\Vert_m\leqslant C_m\Vert\varphi\Vert_m
\end{eqnarray*}
therefore the family $\left(\psi\Phi^{\lambda*}\varphi\right)_\lambda$ is bounded
in the Fr\'echet space $\mathcal{D}_{\tilde{K}}(\mathbb{R}^n\setminus X)$.

The family $\left(\lambda^{-s}\Phi^{\lambda*}t\right)$
is weakly bounded in
$(\mathcal{D}_{\tilde{K}}(\mathbb{R}^n\setminus X))^\prime$
thus strongly bounded by the
uniform boundedness principle since $\mathcal{D}_{\tilde{K}}(\mathbb{R}^n\setminus X)$ is Fr\'echet (\cite[Thm 2.5 p.~44]{Rudin}):
\begin{eqnarray}
\exists C^\prime>0,m\in\mathbb{N}, \forall\lambda ,\forall\varphi\in \mathcal{D}_{\tilde{K}}(\mathbb{R}^n\setminus X),
\vert \left(\lambda^{-s}\Phi^{\lambda*}t\right)(\varphi)\vert\leqslant C^\prime\Vert \varphi\Vert_m. 
\end{eqnarray}
Therefore
\begin{eqnarray*}
\sup_{\lambda\leqslant 2L} \vert \left(\lambda^{-s}\Phi^{\lambda*}t\right)\left(\psi\Phi^{\lambda*}\varphi\right)  \vert
&\leqslant & C^\prime\Vert \psi\Phi^{\lambda*}\varphi \Vert_m\\
&\leqslant & C^\prime C_m \Vert\varphi\Vert_m\\
\implies \vert t(\varphi)\vert &\leqslant & C(1+d(\text{supp }\varphi,X)^{s+n_2}) \Vert\varphi\Vert_m
\end{eqnarray*}
for some $C>0$ independent of $\varphi\in\mathcal{D}_B(\mathbb{R}^n\setminus X)$.
\end{proof}

\section{Renormalized products.}
Let $X\subset \mathbb{R}^n$
be some closed subset.
In this section, we first define
the class $\mathcal{M}(X,\mathbb{R}^n)$ of
tempered functions along $X$:
\begin{defi}\label{tempdefi}
$f$ is tempered along $X$ if
\begin{eqnarray}
\forall m\in\mathbb{N},\forall K\subset \mathbb{R}^n \text{ compact, } 
\exists (C_m,s)\in\mathbb{R}_{\geqslant 0}^{2},
\sup_{\vert\alpha\vert\leqslant m}\vert \partial^\alpha f(x)\vert\leqslant C(1+d(x,X)^{-s}).
\end{eqnarray}
\end{defi}
Tempered functions
form an \textbf{algebra} by Leibniz rule. It is immediate
that the definition \ref{tempdefi}
can be generalized to some closed subset $X$
in a manifold $M$:
we follow the notations of  
the partition of unity argument in 
\ref{Partitionofunityargument}, 
$f$ is tempered along $X$ i.e. 
$f\in\mathcal{M}(X,M)$ if in any
local chart $\psi_i:U_i\subset M \mapsto V\subset\mathbb{R}^n$,
$\psi_{i*}\left(\varphi_if\right)\in\mathcal{M}(\psi_i(X),\mathbb{R}^n)$.

 Then 
we establish a theorem
about renormalized products:
\begin{thm}\label{renormproduct}
Let $M$ be a manifold and $X\subset M$ a closed subset. 
For all $f\in\mathcal{M}(X,M)$ and all $t\in\mathcal{D}^\prime(M)$, 
there exists a distribution
$\mathcal{R}(ft)\in\mathcal{D}^\prime(M)$  
which coincides
with the regular product $ft$ outside
$X$.
\end{thm}

Thanks to the partition of unity argument
of \ref{Partitionofunityargument}, we may reduce
to the case where $X$ is some closed subset of $M=\mathbb{R}^n$
hence $f\in\mathcal{M}(X,\mathbb{R}^n)$
and $t\in\mathcal{E}^\prime(\mathbb{R}^n)$.
By Theorem \ref{euclideanmainthm},
distributions with moderate growth 
are extendible, therefore it suffices
to prove that $ft$ has moderate growth along $X$
which is the content of the following proposition:
\begin{prop}
Let $t\in\mathcal{D}^\prime_K(\mathbb{R}^n\setminus X)$ and 
$f\in C^\infty(\mathbb{R}^n\setminus X)$ 
such that $(t,f)$ satisfy the estimates:
\begin{eqnarray}
\exists (C,s_1)\in\mathbb{R}_{\geqslant 0}^{2},\forall\varphi \in\mathcal{I}(X,\mathbb{R}^n), 
\vert t(\varphi)\vert\leqslant C(1+d(\text{supp }\varphi,X)^{-s_1})\Vert\varphi\Vert_m^K\\
\exists (C_m,s_2)\in\mathbb{R}_{\geqslant 0}^{2},
\forall x\in K\setminus X, \sup_{\vert\alpha\vert\leqslant m}\vert \partial^\alpha f(x)\vert\leqslant C_m(1+d(x,X)^{-s_2}).
\end{eqnarray}
Then $ft$ satisfies
the estimate:
\begin{eqnarray}
\exists C^\prime, \forall\varphi \in\mathcal{I}(X,\mathbb{R}^n), 
\vert ft(\varphi)\vert\leqslant C^\prime(1+d(\text{supp }\varphi,X)^{-(s_1+s_2)})\Vert\varphi\Vert_m^K.
\end{eqnarray}
\end{prop}
\begin{proof}
The claim follows from the estimate:
\begin{eqnarray*}
\forall\varphi \in\mathcal{I}(X,\mathbb{R}^n), \vert ft(\varphi)\vert &\leqslant &C(1+d(\text{supp }\varphi,X)^{-s_1})\Vert f\varphi\Vert_m^K\\
&\leqslant & CC_m2^{mn}(1+d(\text{supp }\varphi,X)^{-s_1})(1+d(\text{supp }\varphi,X)^{-s_2})\Vert\varphi\Vert_m^K\text{ by Leibniz rule }\\
&\leqslant & \underset{C^\prime}{\underbrace{4CC_m2^{mn}}}(1+d(\text{supp }\varphi,X)^{-(s_1+s_2)})\Vert\varphi\Vert_m^K.
\end{eqnarray*}
\end{proof}

\begin{ex}
Our result shares some similarities
with \cite[Theorem 4.3 p.~85]{Meyer-98} where
Meyer renormalizes
the product
of distributions 
$S_\gamma t$ at a point $x_0\in\mathbb{R}^n$, 
$S_\gamma(x)=fp\vert x-x_0\vert^\gamma$ (Hadamard's finite
part), $t$ is a distribution
which is weakly homogeneous of degree $s$ at $x_0$ and
$s+\gamma\notin-\mathbb{N}$. He shows that
the renormalized product $S_\gamma t$ is weakly homogeneous of degree $s+\gamma$
at $x_0$.
\end{ex}
Let us recall that by Theorem \ref{euclideanmainthm}, 
the space $\mathcal{T}_{\mathbb{R}^n\setminus X}(\mathbb{R}^n)$
of distributions with moderate growth along $X$ corresponds with 
the quotient
space
$\mathcal{D}^\prime(\mathbb{R}^n)/\mathcal{D}_X^\prime(\mathbb{R}^n)$
of distributions on $\mathbb{R}^n\setminus X$ extendible
on $\mathbb{R}^n$. Therefore Theorem \ref{renormproduct} implies that:
\begin{coro}\label{tempfunmodules}
$\mathcal{T}_{M\setminus X}(M)$
is a left $\mathcal{M}(X,M)$ module.
\end{coro}
This was also proved
by Malgrange \cite[Proposition 1 p.~4]{Malgrange}.

Let us consider a function $g\in C^\infty(\mathbb{R}^n)$, 
$X=\{g=0\}$
and $gC^\infty(\mathbb{R}^n)$
is a \textbf{closed ideal} of $C^\infty(\mathbb{R}^n)$, then
a result of Malgrange \cite[inequality (2.1) p.~88]{Malgrangeideals}
yields that $g$ satisfies the Lojasiewicz inequality:
\begin{equation}
\forall K\text{ compact },\exists (C,s)\in\mathbb{R}^2_{\geqslant 0},
\forall x\in K, \vert g(x)\vert\geqslant Cd(x,X)^s. 
\end{equation}
It follows by Leibniz rule that $f=g^{-1}$ must be tempered along $X$.
We state and prove a specific case of
''renormalized product'' which is due to Malgrange \cite[Thm 2.1 p.~100]{Malgrangeideals}:
\begin{thm}
Let $M$ be a smooth paracompact manifold, let $f=g^{-1}$, $g\in C^\infty(M)$ such that
the ideal $gC^\infty(M)$ is closed.
Then 
\begin{eqnarray} 
\forall T\in\mathcal{D}^\prime(M),\exists S\in\mathcal{D}^\prime(M) \text{ s.t. }gS=T
\end{eqnarray} 
in particular, $S=fT$ outside $X$.
\end{thm}
Beware that the renormalized
product $S=fT$ is \emph{not uniquely 
defined}, however
it satisfies the equation $gS=T$ whereas
without the closedness assumption on $gC^\infty(M)$,
we would only have $gS=T$ modulo distributions
supported by $X$.
\begin{proof}
By partition of unity,
it suffices to prove that
the linear map
$m_g:t\in \mathcal{E}^\prime(M)\longmapsto gt\in \mathcal{E}^\prime(M)$
is onto if $gC^\infty(M)$ is closed in $C^\infty(M)$.
We will establish that
$m_g$ has closed range and that $ran(m_g)$ is dense
in $\mathcal{E}^\prime(M)$.
 
 $gC^\infty(M)$ is closed in $C^\infty(M)$ implies that
the transposed map:
$m_g^*:C^\infty(M)\longmapsto C^\infty(M)$ has closed range
therefore $m_g$ has closed range since $C^\infty(M)$ is Fr\'echet
and $\mathcal{E}^\prime(M)=C^\infty(M)^\prime$ (see \cite[Thm 26.3 p.~307]{VogtMeise}).
 
$gC^\infty(M)$ is closed in $C^\infty(M)$ hence it is Fr\'echet. 
By the open mapping Theorem \cite[Thm 8.5 p.~60]{VogtMeise}, 
$m_g:C^\infty(M)\mapsto gC^\infty(M)$ is a linear continuous, surjective map of Fr\'echet spaces
hence $m_g$ is \textbf{open}. In terms
of estimates, this implies that for any continuous
seminorm $\Vert.\Vert_{m}^K$ of $C^\infty(M)$, 
there is a continuous seminorm $\Vert.\Vert_{m^\prime}^{K^\prime}$ 
such that
$\Vert\varphi\Vert_{m}^K\leqslant \Vert(g\varphi)\Vert_{m^\prime}^{K^\prime}$
(see \cite[inequality (2.2) p.~88]{Malgrangeideals}), 
hence $g\varphi=0\implies \varphi=0$.
Then we conclude by the observation that
$ran(m_g)^\perp=\{\varphi\in C^\infty(M) \text{ s.t. }\forall t\in\mathcal{E}^\prime(M),
gt(\varphi)=0 \}=\{ \varphi \text{ s.t. }g\varphi=0 \}=\{0\}\implies ran(m_g)$ 
is everywhere dense in $\mathcal{E}^\prime(\mathbb{R}^n)$.
\end{proof}
\section{Renormalization of Feynman amplitudes in Euclidean 
quantum field theories.}
\subsection{Feynman amplitudes are extendible.}
We give the main application
of our extension techniques.
Our approach to renormalization
follows the philosophy of Brunetti--Fredenhagen \cite{Brunetti2,BFK,Brunetti}, Nikolov--Stora--Todorov \cite{NST}
which goes back to \cite{Epstein,EGS},
and is based
on the concept of extension of distributions.
However, we will use the  
beautiful formalism of \emph{renormalization maps}
of N. Nikolov \cite{NST,Nikolov} 
which is closest in spirit to the present paper.
In what follows, we will always 
assume that $(M,g)$ is a
smooth $d$-dimensional Riemannian manifold
with Riemannian metric $g$.
We denote by $\Delta_g$ the Laplace
Beltrami operator corresponding to $g$,
and we consider the \emph{Green function}
$G\in\mathcal{D}^\prime(M\times M)$ 
of the operator $\Delta_g+m^2,m\in\mathbb{R}_{\geqslant 0}$. 
$G$
is the Schwartz kernel of the operator
inverse of $\Delta_g+m^2$ (\cite[Appendix 1]{ShubinNantes})
which always exists when $M$ is compact and $m^2\notin \text{Spec}(\Delta_g)$. 
In the noncompact case, the 
general existence and uniqueness result for the Green function 
usually depends on the global properties
of $\Delta_g$ and $(M,g)$. If $(M,g)$ 
has
\emph{bounded geometry}
in the sense of \cite[p.~33]{CGT} and \cite{Roeindex}
(see also \cite[Definition 1.1 Appendix 1]{ShubinNantes},\cite[Def 1.1 p.~3]{ShubinBloch}),
then
one can find
in \cite[Appendix 1]{ShubinNantes} conditions
of spectral theoretic nature on $\Delta_g,m^2$
that imply the existence of an operator
inverse $\left(\Delta_g+m^2 \right)^{-1}:L^p(M)\mapsto L^p(M),p\in(1,+\infty)$
whose Schwartz kernel is $G$.

However if $G$ exists, then we show
a fundamental result
about the asymptotics 
of $G$ near the diagonal:
\begin{lemm}
Let $(M,g)$ be a smooth Riemannian manifold
and $\Delta_g$ the corresponding
Laplace operator.
If
$G\in\mathcal{D}^\prime(M\times M)$ is the fundamental solution of
$\Delta_g+m^2$, then $G$
is tempered along $D_2\subset M^2$.
\end{lemm}
\begin{proof}
Temperedness is a local property
therefore it suffices to prove the Lemma for
some compact
domain $\Omega\times\Omega\subset\mathbb{R}^d\times\mathbb{R}^d$
and $g$ is a Riemannian metric on $\mathbb{R}^d$. 
The differential operator $\Delta_g+m^2$ 
is elliptic
with smooth coefficients, $G$ is a fundamental solution
of $\Delta_g+m^2$ in particular it is a \emph{parametrix} of $\Delta_g+m^2$
which implies it is
an elliptic pseudodifferential operator with polyhomogeneous symbol
\cite[Thm 2.7 p.~55]{Shimakura}. 
Set $\mathcal{E}(x,z)=G(x,x+z)$, then
by \cite[Theorem 3.3 p.~58]{Shimakura},
there exists two sequences $(\mathcal{A}_q(x,z))_q,(\mathcal{B}_q(x,z))_q$
of functions smooth on $\Omega$ w.r.t. $x$ 
and real analytic
on $\mathbb{S}^{d-1}$ w.r.t. $z$ such that
$\mathcal{E}$
satisfies the following estimate 
(which is adapted from \cite[(3.14) p.~59]{Shimakura}):
there is some $l$ s.t.
for every  multi-index $\alpha,\beta$, there exists $N\in\mathbb{N},c\in\mathbb{R}$ 
s.t.
\begin{eqnarray}
\vert\partial^\alpha_x\partial^\beta_z\mathcal{E}(x,z) \vert\leqslant c+
\vert \partial^\alpha_x\partial^\beta_z \sum_{q=0}^N 
\left(\vert z\vert^{l+q-d}\left(\mathcal{A}_q(x,\frac{z}{\vert z\vert})\log\vert z\vert+\mathcal{B}_q(x,\frac{z}{\vert z\vert})\right)\right) \vert 
\end{eqnarray} 
if $x\in \Omega,\vert z\vert\leqslant 1$.
The right hand side has moderate growth
along $\{z=0\}$ and $\mathcal{E}(x,z)$ 
is thus tempered along $\{z=0\}$ which implies that
$G(x,y)$ is tempered along $D_2$.
\end{proof}

\paragraph{Configuration spaces.}

For every finite subset $I\subset \mathbb{N}$ and open subset $U\subset M$, we define
the configuration space
$U^I=\text{Maps }(I\mapsto U)=\{(x_i)_{i\in I}\text{ s.t. }x_i\in U,\forall i\in I \} $
of $\vert I\vert$ particles in $U$ 
labelled by the subset $I\subset\mathbb{N}$.
In the sequel, we will distinguish two types of diagonals
in $U^I$, the \emph{big diagonal} 
$D_I=\{(x_i)_{i\in I} \text{ s.t. }\exists (i\neq j)\in I^2, x_i=x_j \}$
which represents configurations where
at least two particles
collide, and the \emph{small diagonal}
$d_I=\{(x_i)_{i\in I} \text{ s.t. } \forall (i,j)\in I^2, x_i=x_j  \}$
where all particles in $U^I$ collapse over the same element.
The configuration space $M^{\{1,\dots,n\}}$ and the corresponding \emph{big and small} diagonals $D_{\{1,\dots,n\}},d_{\{1,\dots,n\}}$ will be denoted by
$M^n,D_n,d_n$ for simplicity.
\begin{thm}\label{extensionfeynmanamplitudes}
Let $(M,g)$ be a smooth Riemannian manifold, 
$\Delta_g$ the corresponding
Laplace operator and $G$ the Green
function of $\Delta_g+m^2$. 
Then all ''Feynman amplitudes'' of the form:
\begin{equation}
\prod_{1\leqslant i<j\leqslant n} G^{n_{ij}}(x_i,x_j)\in C^\infty(M^n\setminus D_n), n_{ij}\in\mathbb{N}
\end{equation}
are \textbf{extendible} in $\mathcal{D}^\prime(M^n)$.
\end{thm}
\begin{proof}
$d_{ij}\subset D_n\implies \forall s\geqslant 0, d(x,d_{ij})^{-s}\leqslant
d(x,D_n)^{-s}$ together with the fact that
$G(x_i,x_j)$ is tempered
along $d_{ij}$ imply that 
$G(x_i,x_j)\in\mathcal{M}(D_n,M^n)$.
Since $\mathcal{M}(D_n,M^n)$ is an \textbf{algebra}, $\prod_{1\leqslant i<j\leqslant n} G^{n_{ij}}(x_i,x_j)\in 
\mathcal{M}(D_n,M^n)$ and is therefore extendible
on $M^n$ by Theorem \ref{renormproduct}.
\end{proof}

\subsection{Renormalization maps, locality and the factorization property.}

\paragraph{The vector subspace $\mathcal{O}(D_I,.)$ generated
by Feynman amplitudes.}\label{defifeynmanamplitudesmodule}

In QFT, renormalization is not only extension of Feynman amplitudes
in configuration space but our extension procedure
should satisfy some consistency conditions in order to
be compatible with the fundamental requirement of \textbf{locality}.

Recall that for any open subset $\Omega\subset M^I$, we denote
by $\mathcal{M}(D_I,\Omega)$ the \textbf{algebra} of tempered functions
along $D_I$. 
We introduce the vector space $\mathcal{O}(D_I,\Omega)\subset \mathcal{M}(D_I,\Omega)$
generated by the Feynman amplitudes
\begin{equation}
\mathcal{O}(D_I,\Omega)=\left\langle\left( \prod_{i<j\in I^2} G^{n_{ij}}(x_i,x_j)\right)_{n_{ij}}\right\rangle_{\mathbb{C}}.
\end{equation}
\paragraph{Axioms for renormalization maps: factorization property as a consequence of locality.}

We define a collection of \emph{renormalization maps} 
$\left(\mathcal{R}_{\Omega\subset M^I}\right)_{\Omega,I}$ 
where $I$ runs over the finite subsets of $\mathbb{N}$
and $\Omega$ runs over the open subsets of $M^I$
which satisfy the following
axioms which are simplified versions 
of those figuring in \cite[2.3 p.~12--14]{Nikolov} \cite[Section 5 p.~33--35]{NST}:

\begin{defi}\label{axiomsrenormmaps}

\begin{enumerate}
\item For every $I\subset \mathbb{N},\vert I\vert<+\infty$, $\Omega\subset M^I$,
$\mathcal{R}_{\Omega\subset M^I}$ is a \textbf{linear extension operator}:
\begin{equation}
\mathcal{R}_{\Omega\subset M^I}:\mathcal{O}(D_I,\Omega)\longmapsto \mathcal{D}^{\prime}(\Omega).
\end{equation} 
\item For all inclusion of open subsets
$\Omega_1\subset\Omega_2\subset M^I$, we require that: 
\begin{eqnarray*}
\forall f\in\mathcal{O}(D_I,\Omega_2),\forall\varphi\in\mathcal{D}(\Omega_1)  \\
\left\langle \mathcal{R}_{\Omega_2\subset M^I}(f),\varphi\right\rangle
=\left\langle\mathcal{R}_{\Omega_1\subset M^I}(f),\varphi \right\rangle.
\end{eqnarray*}
\item The renormalization maps satisfy the \textbf{factorization property}.
If $(U,V)$ are disjoint open subsets of $M$, and $(I,J)$ are
disjoint finite subsets of $\mathbb{N}$, 
$\forall (f,g)\in \mathcal{O}(D_I,U^I)\times \mathcal{O}(D_J,V^J)$ :
\begin{eqnarray}
\mathcal{R}_{(U^I\times V^J)\subset M^{I\cup J}}(f\otimes g)&=&\underset{\in \mathcal{D}^\prime(U^I)}{\underbrace{\mathcal{R}_{U^I\subset M^I}(f)}}\otimes \underset{\in \mathcal{D}^\prime(V^J)}{\underbrace{\mathcal{R}_{V^J\subset M^J}(g)}}\in \mathcal{D}^\prime(U^I\times V^J)
\end{eqnarray}
\end{enumerate}
\end{defi}

The most important property is the factorization property $(3)$ which is imposed in
\cite[equation (2.2) p.~5]{NST}.

\paragraph{Remarks on the axioms of the Renormalization maps.}

To define $\mathcal{R}$ on $M^I$, 
it suffices to define $\mathcal{R}_{\Omega_i\subset M^I}$ 
for an open cover $(\Omega_i)_i$
of $M^I$ (they do not necessarily coincide on the overlaps $\Omega_i\cap\Omega_j$)
and glue the determinations by a partition of unity. 
\paragraph{Uniqueness property of renormalization maps.}
The following Lemma is proved in \cite[Lemmas 2.2, 2.3 p.~6]{NST}
and tells us that if a collection
of renormalization maps $(\mathcal{R}_{\Omega\subset M^I})_{\Omega,I}$ exists 
and satisfies the list of axioms \ref{axiomsrenormmaps}
then the restriction of $\mathcal{R}_{M^n}(\prod_{1\leqslant i<j\leqslant n} G^{n_{ij}}(x_i,x_j))$
on $M^n\setminus d_n$ would be 
uniquely determined
by the renormalizations $\mathcal{R}_{M^I}$ for all $\vert I\vert<n$ because of
the factorization axiom.
\begin{lemm}\label{keylemmaNST}
Let $(\mathcal{R}_{\Omega\subset M^I})_{\Omega,I}$ be a collection
of renormalization maps satisfying the axioms
\ref{axiomsrenormmaps}. Then for any Feynman amplitude
$\prod_{1\leqslant i<j\leqslant n} G^{n_{ij}}(x_i,x_j)$,
the renormalization
$\mathcal{R}_{M^n\setminus d_n\subset M^n}
(\prod_{1\leqslant i<j\leqslant n} G^{n_{ij}}(x_i,x_j))$ is uniquely determined
by the renormalizations 
$\mathcal{R}_{M^I}(\prod_{i<j\in I^2} G^{n_{ij}}(x_i,x_j))$ 
for all $\vert I\vert<n$.
\end{lemm}
\begin{proof}
See \cite[p.~6-7]{NST} for the detailed proof.
\end{proof}

Beware that the above Lemma
\textbf{does not imply the existence} of 
renormalization maps but only that
they must satisfy certain 
consistency conditions if they exist.
\subsection{The existence Theorem for renormalization maps.}
Now we give a short proof 
of the
existence of renormalization maps
on general Riemannian manifolds.
Recall $(M,g)$ is a smooth Riemannian manifold, 
$\Delta_g$ the corresponding
Laplace operator, $G$ the Green
function of $\Delta_g+m^2$ and for any open subset $\Omega\subset M^I$,
$\mathcal{O}(D_I,\Omega)$ is the vector space
generated by the Feynman amplitudes.
\begin{thm}\label{renormthmqft}
There exists
a collection of renormalization maps
$\left(\mathcal{R}_{\Omega\subset M^I}\right)_{\Omega,I}$ 
where $I$ runs over the finite subsets of $\mathbb{N}$
and $\Omega$ runs over the open subsets of $M^I$
which satisfies the axioms \ref{axiomsrenormmaps}.
\end{thm}
\begin{proof}
We proceed by induction on $n$. 

It 
suffices to establish the existence
of $\mathcal{R}_{M^n}\left(\prod_{1\leqslant i<j\leqslant n} G^{n_{ij}}(x_i,x_j)\right)$
for generic Feynman amplitudes
$\prod_{1\leqslant i<j\leqslant n} G^{n_{ij}}(x_i,x_j)\in\mathcal{O}(D_n,M^n)$.

\textbf{Step 1},
we initialize our induction
with $\mathcal{R}_{M^2}:\mathcal{O}\left(D_2, M^2\right)\mapsto \mathcal{D}^\prime(M^2)$ 
whose existence is guaranteed by Theorem \ref{extensionfeynmanamplitudes}.

\textbf{Step 2}, by Lemma \ref{coveringlemmma},
the complement $M^n\setminus d_n$ of the
\emph{small diagonal} $d_n$ in $M^n$ 
is covered by open sets
of the form
$C_{IJ}=M^n\setminus\left(\cup_{(i,j)\in I\times J}d_{ij}\right)$ where
$I\cup J=\{1,\dots,n\}, I\cap J=\emptyset$. In the sequel, we
write $\mathcal{R}_{C_{IJ}}$ instead of 
$\mathcal{R}_{C_{IJ}\subset M^n}$ for simplicity.

By factorization property, we also find that:
\begin{eqnarray*}
\mathcal{R}_{C_{IJ}}\left(\prod_{1\leqslant i<j\leqslant n} G^{n_{ij}}(x_i,x_j)\right)&=&\underset{\in\mathcal{D}^\prime(M^I)}{\underbrace{\mathcal{R}_{M^I}\left(G_I\right)}}
\underset{\in\mathcal{D}^\prime(M^J)}{\underbrace{\mathcal{R}_{M^J}\left(G_J\right)}}
\underset{\in \mathcal(\partial C_{IJ},M^n)}{\underbrace{\prod_{(i,j)\in I\times J} G^{n_{ij}}(x_i,x_j)}}\\
G_I=\prod_{(i<j)\in I^2} G^{n_{ij}}(x_i,x_j),&&
G_J=\prod_{(i<j)\in J^2} G^{n_{ij}}(x_i,x_j)
\end{eqnarray*}
therefore the renormalization map 
$\mathcal{R}_{M^n\setminus d_n}$ is entirely
determined by the
renormalization maps $\mathcal{R}_{M^I}$
for $\vert I\vert\leqslant n-1$ and the determination
is in fact \textbf{unique} according
to Lemma \ref{keylemmaNST}.

\textbf{Step 3}, in Lemma \ref{temperedpartitionofunity}, we construct a partition of unity $(\chi_{IJ})_{IJ}$ of
$M^n\setminus d_n$ subordinated to the open cover $(C_{IJ})_{IJ}$ i.e.
$\text{supp }\chi_{IJ}\subset C_{IJ} ,\sum_{IJ}\chi_{IJ}=1$ such that
each $\chi_{IJ}$ satisfies the essential property of being 
\textbf{tempered along} $d_n$. 

\textbf{Step 4}, the key idea is that
the product
$\mathcal{R}_{M^I}\left(G_I\right)\mathcal{R}_{M^J}\left(G_J\right)$ is well defined
in $\mathcal{D}^\prime(M^n)$ and the product $\prod_{(i,j)\in I\times J} G^{n_{ij}}(x_i,x_j)$
is tempered along $\partial C_{IJ}$.
Therefore 
\begin{eqnarray*}
\chi_{IJ}\mathcal{R}_{C_{IJ}}(\underset{1\leqslant i<j\leqslant n}{\prod} G^{n_{ij}}(x_i,x_j))=\underset{\in \mathcal{M}(\partial C_{IJ},M^n)}{\underbrace{\chi_{IJ}\prod_{(i,j)\in I\times J} G^{n_{ij}}(x_i,x_j)}} \underset{\in\mathcal{D}^\prime(M^n)}{\underbrace{\mathcal{R}_{M^I}\left(G_I\right)\mathcal{R}_{M^J}\left(G_J\right)}}\in\mathcal{D}^\prime(C_{IJ})
\end{eqnarray*}
is a product of tempered
functions along $\partial C_{IJ}$ with a distribution in $\mathcal{D}^\prime(M^n)$,
therefore it has an extension $\overline{\chi_{IJ}\mathcal{R}_{C_{IJ}}(\underset{1\leqslant i<j\leqslant n}{\prod} G^{n_{ij}}(x_i,x_j))}$ in $\mathcal{D}_{\overline{C_{IJ}}}^\prime(M^n)$ by Theorem \ref{renormproduct}. By construction, $\chi_{IJ}$ vanishes in some neighborhood
of $\partial C_{IJ}\setminus d_n$ in $M^n\setminus d_n$ which implies that $\chi_{IJ}\mathcal{R}_{C_{IJ}}(\underset{1\leqslant i<j\leqslant n}{\prod} G^{n_{ij}}(x_i,x_j))=\overline{\chi_{IJ}\mathcal{R}_{C_{IJ}}(\underset{1\leqslant i<j\leqslant n}{\prod} G^{n_{ij}}(x_i,x_j))}$
in $\mathcal{D}^\prime(M^n\setminus d_n)$.
Then we define $\mathcal{R}_n(\underset{1\leqslant i<j\leqslant n}{\prod} G^{n_{ij}}(x_i,x_j))$ to be the distribution
\begin{equation}
\sum_{IJ}\overline{\chi_{IJ}\mathcal{R}_{C_{IJ}}(\underset{1\leqslant i<j\leqslant n}{\prod} G^{n_{ij}}(x_i,x_j))}.
\end{equation}

\end{proof}
\paragraph{Covering lemma.}
The following Lemma is due to Popineau and Stora \cite[Lemma 2.2 p.~6]{NST}
\cite{Stora02, Popineau} and states that
$M^n\setminus d_n$ can be partitioned as a union of open sets
on which the renormalization map $\mathcal{R}_n$ can factorize. 
\begin{lemm}\label{coveringlemmma}
Let $M$ be a smooth manifold and for all finite subsets $(I,J)$ of $\mathbb{N}$
s.t. $ I\cap J=\emptyset, I\cup J=\{1,\dots,n\}$, let $C_{IJ}=\{(x_1,\dots,x_n)\text{ s.t. }\forall (i,j)\in I\times J x_i\neq x_j\}\subset M^n$. Then 
\begin{eqnarray}
\underset{I\cap J=\emptyset, I\cup J=\{1,\dots,n\}}{\bigcup}C_{IJ}=M^n\setminus d_n.
\end{eqnarray} 
\end{lemm}
\begin{proof}
The key observation is the following,
$(x_1,\dots,x_n)\in d_n\Leftrightarrow \forall U \text{ neighborhood of }x_1, (x_1,\dots,x_n)\in U^n$.
On the contrary 
\begin{eqnarray*}
&&(x_1,\dots,x_n)\notin d_n\\
&\Leftrightarrow & \exists (U,V) \text{ open s.t. }\overline{U}\cap \overline{V}=\emptyset, I\cup J=\{1,\dots,n\}, I\cap J=\emptyset \text{ s.t. } (x_1,\dots,x_n)\in U^I\times V^J.
\end{eqnarray*}
It suffices to set $\varepsilon=\underset{1<i\leqslant n }{\inf}\{d(x_i,x_1) \text{ s.t. } d(x_i,x_1)>0\}$ then let $U=\{x \text{ s.t. }d(x,x_1)< \frac{\varepsilon}{3}\}$ and $V=\{x \text{ s.t. }d(x,x_1)>\frac{2\varepsilon}{3}\}$.

It follows that the complement $M^n\setminus d_n$ of the
\emph{small diagonal} $d_n$ in $M^n$ 
is covered by open sets
of the form
$C_{IJ}=M^n\setminus\left(\cup_{(i,j)\in I\times J}d_{ij}\right)$ where
$I\cup J=\{1,\dots,n\}, I\cap J=\emptyset$.
\end{proof}
\paragraph{Tempered partition of unity
associated to the cover $\left(C_{IJ}\right)_{IJ}$.}
\begin{lemm}\label{temperedpartitionofunity}
Let $M$ be a smooth manifold and let $(C_{IJ})_{IJ}$
be the cover of $M^n\setminus d_n$ defined in Lemma \ref{coveringlemmma}
then there exists a partition of unity
$(\chi_{IJ})_{IJ}$ subordinated to $(C_{IJ})_{IJ}$
such that every function $\chi_{IJ}$ is tempered along $d_n$.
\end{lemm}
\begin{proof}
We will first construct a partition of unity in some
neighborhood $\mathcal{N}$ of $d_n$. 
Consider the normal bundle $N(d_n\subset M^n)$ 
of $d_n$ in $M^n$, 
by the tubular neighborhood Theorem, 
there is some neighborhood
$\mathcal{N}$ of $d_n$ in $M^n$ and
a diffeomorphism $\Phi:\mathcal{N}\mapsto \Phi(\mathcal{N})\subset  N(d_n\subset M^n)$ 
which identifies
the neighborhood $\Phi(\mathcal{N})$ of the zero section $\underline{0}\subset N(d_n\subset M^n)$
with $\mathcal{N}\subset M^n$.

Let us denote by $\Phi(\mathcal{N})\setminus \underline{0}$
the neighborhood of the zero section $\underline{0}$ deprived of $\underline{0}$.
Then $\Phi(C_{IJ})_{IJ}$ forms a conical open cover
of $\Phi(\mathcal{N})\setminus \underline{0}$. To see what happens in coordinates,
let $U_i$ be some cover of $M$,
we trivialize our bundle over $U_i$, 
$\Psi_{U_i}:N(d_n\subset M)|_{U_i}\mapsto U_i\times 
\mathbb{R}^{(n-1)d}$ 
and we use local coordinates $(x, h_2,\dots,h_n)\in U_i\times\mathbb{R}^{n(d-1)}$. 
The set
$\Psi_{U_i}\circ\Phi\left(C_{IJ}\right)$ has simple expression $\Psi_{U_i}\circ\Phi\left(C_{IJ}\right)=\underset{(i,j)\in I\times J}{\cap}  \{h_i-h_j\neq 0\}$ and is 
\textbf{therefore invariant under scalings
$(x,h_2,\dots,h_n)\mapsto (x,\lambda h_2,\dots,\lambda h_n) $ and does not depend
on $x\in U_i$}. Therefore the sets
$\left(\Psi_{U_i}\circ\Phi\left(C_{IJ}\right)\right)_{IJ}$ is an open conical cover
of $\left(\mathbb{R}^{d(n-1)}\setminus \{0\}\right)\times U_i$. 
Let $\left(\chi_{iIJ}\right)_{IJ}$ be the corresponding partition of unity, then
we can choose every $\chi_{iIJ}$ of the form $f_{iIJ}(\frac{h}{\vert h\vert}),\vert h\vert=\sqrt{\sum_{i=2}^n h_i^2},f_{iIJ}\in C^\infty(\mathbb{R}^{d(n-1)}\setminus \{0\})$
since $\Psi_{U_i}\circ\Phi\left(C_{IJ}\right)$ is conical and does not depend on $x$. 
Therefore
combining the Faa Di Bruno 
formula and the fact 
that 
$\vert\partial^k_h \left(\frac{h}{\vert h\vert}\right)\vert\leqslant C_k(1+\vert h\vert^{-\vert k\vert})$
yields that:
\begin{eqnarray*}
\vert \partial^\alpha_h \chi_{iIJ}(h)\vert=\vert\partial^\alpha_h f_{iIJ}(\frac{h}{\vert h\vert }) \vert\leqslant
C_\alpha(1+\vert h\vert^{-\vert\alpha\vert} ) 
\end{eqnarray*}
which implies that
$\chi_{iIJ} $ is tempered along $U_i\times\{0\}$ therefore $\Psi_{U_i}^*\chi_{iIJ}$
is tempered along
the zero section $\underline{0}|_{U_i}\subset N(d_n\subset M^n)|_{U_i}$.
Let $(\varphi_i)_i$ be a partition of unity subordinated 
to the cover $(U_i)_i$ of $M$, then the functions
$\left(\sum_{i}\varphi_i\psi_i^*\chi^i_{IJ}\right)_{IJ}$ form a 
partition of unity
of $N(d_n\subset M^n)\setminus\{\underline{0}\}$
which is subordinated to the conic cover $\Phi(C_{IJ})_{IJ}$.

To go back to the configuration space $M^n$, choose a neighborhood $\mathcal{N}^\prime$ of $d_n$ 
s.t. $\mathcal{N}$ is a neighborhood of $\mathcal{N}^\prime$, we have the inclusions
$d_n\subset \mathcal{N}^\prime\subset \mathcal{N}$.
Let $\chi_1,\chi_2$ be a partition of unity
subordinated to the cover $(\mathcal{N},M^n\setminus \mathcal{N}^\prime)$
and choose $\left(\tilde{\chi}_{IJ}\right)_{IJ}$ to be an arbitrary partition of unity subordinated
to the cover $\left(C_{IJ}\right)_{IJ}$ of $M^n\setminus d_n$.
Then set $\chi_{IJ}=\chi_1\Phi^*\left(\sum_{i}\varphi_i\Psi_{U_i}^*\chi^i_{IJ}\right)+\chi_2\tilde{\chi}_{IJ}$ 
and it 
follows by construction that every $\chi_{IJ}$ is tempered along $d_n$.
\end{proof}

\end{document}